\documentclass[11pt]{article}
\usepackage{amsmath,amsfonts,amsthm,amssymb,verbatim}

\usepackage{xcolor}
\usepackage{enumerate}
\usepackage[left=1in, top=1in, bottom=1in, right=1in]{geometry}
\usepackage{suffix}
\usepackage[small]{titlesec}

\newenvironment{ctable*}
{\begin{table*}[t]\begin{center}}{\end{center}\end{table*}}

\newtheorem{theorem}{Theorem}[section]

\newtheorem{lemma}[theorem]{Lemma}

\newtheorem{proposition}[theorem]{Proposition}

\theoremstyle{definition}
\newtheorem{definition}[theorem]{Definition}

\newcommand{\Xomit}[1]{ }

\newcommand{\OPT}{\ensuremath{\mathit{OPT}}}

\newcommand{\sig}{\varphi}

\renewcommand{\vec}[1]{\boldsymbol{#1}}
\newcommand{\ind}{\mathbf{1}}
\newcommand{\norm}[1]{\lVert #1 \rVert}
\newcommand{\ip}[1]{\langle #1 \rangle}
\newcommand{\ct}{c}
\newcommand{\niggly}[1]{\eta(\vec{#1})}
\newcommand{\Csmith}{C^{SR}}
\newcommand{\Crand}{C^{R}}
\newcommand{\Cps}{C^{PS}}
\newcommand{\Capr}{C^{A}}
\newcommand{\csmith}{c^{SR}}
\newcommand{\cps}{c^{PS}}
\newcommand{\capr}{c^{A}}
\newcommand{\crand}{c^{R}}
\newcommand{\Macs}{I}
\newcommand{\Jobs}{J}

\newcommand{\p}[1]{\doproctime#1\relax}
\def\doproctime#1,#2\relax{p_{#1 #2}}

\newcommand{\pow}[1]{\dopow#1\relax}
\def\dopow#1,#2\relax{\rho_{#1 #2}}

\def\shf{\textsf{ShortestFirst}}
\def\lof{\textsf{LongestFirst}}

\def\pro{\textsf{ProportionalSharing}}
\def\equ{\textsf{EqualSharing}}
\def\smi{\textsf{SmithRule}}
\def\apr{\textsf{Approx}}
\newcommand{\Rand}{\textsf{Rand}}

\newcommand{\rand}{\Rand}

\newcommand\schedprob[2]{\ensuremath{#1|\,|#2}}
\def\unrel{\schedprob{R}{C_{\max}}}

\def\ident{\schedprob{P}{C_{\max}}}
\def\unrels{\schedprob{R}{\sum w_j\ct_j}}

\def\bipars{\schedprob{B}{\sum w_j\ct_j}}
\def\idents{\schedprob{P}{\sum w_j\ct_j}}

\newcommand\stv{\vec{x}}
\WithSuffix\newcommand\stv*{\vec{x^*}}
\newcommand{\st}{x}
\WithSuffix\newcommand\st*{x^*}
\newcommand{\St}{X}
\WithSuffix\newcommand\St*{X^*}
\WithSuffix\newcommand\stv'{\vec{x'}}

\newcommand{\randterm}[1]{\dorandterm#1\relax}
\def\dorandterm#1,#2,#3\relax{\frac{\pow{#1,#2}\pow{#1,#3}}{\pow{#1,#2} + \pow{#1,#3}}}
\newcommand{\smithterm}[1]{\dosmithterm#1\relax}
\def\dosmithterm#1,#2,#3\relax{\min\{\pow{#1,#2}, \pow{#1,#3}\}}

\newcommand{\Reals}{\mathbb R}
\newcommand{\R}{\Reals}
\newcommand{\Rplus}{\R_+}

\newcommand{\Naturals}{\mathbb N}
\newcommand{\N}{\Naturals}
\newcommand{\poa}{price of anarchy}

\renewcommand{\Pr}{\mathbb{P}}

\newcommand{\PotPS}{\Phi^{PS}}
\newcommand{\PotR}{\Phi^R}
\newcommand{\PotA}{\Phi^A}

\newcommand{\PotD}{D}

\def\clap#1{\hbox to 0pt{\hss#1\hss}}

\def\mathrlap{\mathpalette\mathrlapinternal}

\def\mathrlapinternal#1#2{\rlap{$\mathsurround=0pt#1{#2}$}}

\newcommand{\Sz}{\kappa}

\newcommand{\sparagraph}[1]{\medskip \noindent \emph{#1}}

\begin{document}

\title{ {\bf Inner Product Spaces for MinSum Coordination Mechanisms}\thanks{This work was supported in part by NSF grant CCF0830516, by FONDECYT grant 1090050, and by Andreas Mentzelopoulos Scholarships for the University of Patras.}}

\author{ Richard Cole
         \thanks{{\tt cole@cims.nyu.edu}. Courant Institute, New York University, N.Y.}
        \and
		Jos\'e R. Correa
		\thanks{\texttt{jcorrea@dii.uchile.cl}.  Departamento de Ingenier\'{\i}a Industrial, Universidad de Chile.}
		\and 
         Vasilis Gkatzelis 
         \thanks{{\tt gkatz@cims.nyu.edu}. Courant Institute, New York University, N.Y.}
        \and
         Vahab Mirrokni
        \thanks{{\tt mirrokni@google.com}. Google Research, New York, N.Y.}
		\and 
		Neil Olver
		\thanks{\texttt{olver@math.mit.edu}. Department of Mathematics, MIT.}
    }

\date{}

\maketitle

\begin{abstract} 
We study policies aiming to minimize the weighted sum of completion
times of jobs in the context of coordination mechanisms for selfish scheduling
problems. Our goal is to design local policies that achieve a good price of anarchy in
the resulting equilibria for unrelated machine scheduling.
To obtain the approximation bounds, we introduce a new technique that while conceptually simple, seems to be quite powerful.
The method entails mapping strategy vectors into a carefully chosen inner product space; costs are shown to correspond to the norm in this space, and the Nash condition also has a simple description.
With this structure in place, we are able to prove a number of results, as follows.

First, we consider Smith's Rule, which orders the jobs on a machine in ascending processing time to weight ratio, and show that it achieves an approximation ratio of $4$.
We also demonstrate that this is the best possible for deterministic non-preemptive strongly local policies.
Since Smith's Rule is always optimal for a given fixed assignment, this may seem unsurprising,
but we then show that better approximation ratios can be obtained if either preemption or randomization is allowed.

We prove that \pro, a preemptive strongly local policy, achieves an approximation ratio of $2.618$ for the weighted sum of completion times, and an approximation ratio of $2.5$ in the unweighted case. Again, we observe that these bounds are tight. 
Next, we consider \rand, a natural non-preemptive but \emph{randomized} policy.
We show that it achieves an approximation ratio of at most $2.13$; moreover, if the sum of the weighted completion times is negligible compared to the cost of the optimal solution, this improves to $\pi/2$.

Finally, we show that both \pro\ and \rand\ induce potential games, and thus always have a pure Nash equilibrium (unlike Smith's Rule).
This also allows us to design the first \emph{combinatorial} constant-factor approximation algorithm minimizing weighted completion time 
for unrelated machine scheduling.
It achieves a factor of $2+\epsilon$ for any $\epsilon > 0$, and involves imitating best response dynamics using a variant of \pro\ as the policy.

\end{abstract}

\setcounter{page}{0} \thispagestyle{empty}

\newpage
\setcounter{page}{1}


\section{Introduction}

Traditionally, work in operations research has focused on finding
globally optimal solutions for optimization problems.
In tandem, computer scientists have long studied the effects of a lack of
different kinds of resources, mainly the lack of computational
resources in optimization.
In designing massive decentralized systems, the \emph{lack of coordination} among
different participating agents has become an important consideration. 
This issue is typically addressed
through distributed algorithms in which a central authority designs 
mechanisms (protocols) specifying the rules of the game,
with the goal that the independent and selfish
choices of the users result in a socially desirable outcome.  To
measure the performance of these algorithms, the global objective
function (social cost) is evaluated at equilibrium points for selfish
users. For games, probably the most accepted such measure is the \emph{price of
anarchy}~\cite{KP99}, the worst case ratio of the social cost at a
Nash equilibrium to that at a social optimum;
the same measure can be used for coordination mechanisms; sometimes we call this
their \emph{approximation factor} to highlight that this is a distinct
measure. 

The by now standard approach to bound the price of anarchy (PoA) when social cost is
taken to be the sum of individual costs works as follows \cite{R09}.
First, the social cost is bounded by using the equilibrium conditions,
noting that an individual is better off at equilibrium than she would
be if she unilaterally changed her strategy to the one she would use in a
centralized optimum. Second, the actual (weighted) sum of player costs
is also upper bounded, using an appropriately chosen inequality, by a
linear combination of the social cost of the equilibrium and the
social cost of an optimal solution. 

In this paper we establish a
methodology to deal with the second step of
this proof scheme. Our method interprets the sum in the second step as
an inner product on a suitable space. Then, we apply the
Cauchy-Schwartz inequality in the chosen inner product space, and go
back to the original space by applying a minimum norm distortion
inequality. Many of the existing results employ a special case of this
approach in which the costs can be expressed in terms of quadratic
polynomials to which the Cauchy-Schwartz inequality can be applied
directly without the need for an intermediate inner product space. 
We apply our new method in the context of scheduling jobs on unrelated machines.
Our method elucidates the hidden structure in the games we consider. 
Once the framework has been set up, our proofs become short and elegant, thus we anticipate that this method may prove useful elsewhere too. 

Specifically, we consider the classic problem of scheduling $n$ jobs
on $m$ unrelated machines from a game theoretic perspective. In this
situation, job $j$ takes time $\p{i,j}$ if processed on machine $i$,
and also has an associated weight $w_j$. Although the central goal is
to minimize the weighted sum of completion times of jobs, we consider
the \emph{scheduling game} in which each job is a fully informed player
wanting to minimize its individual weighted completion time, while each
machine announces a policy which it will follow in processing the jobs
it is assigned. Our goal is to choose the policy so as to minimize the
approximation ratio of the actual costs under this policy to the
optimal costs obtainable under any policy.
To this end, several approaches imposing incentives on self-interested
agents have been proposed, including some using monetary
transfers~\cite{Beckman56,CDR06,FJM04,CKK}, and others enforcing
strategies on a fraction of users as a Stackelberg
strategy~\cite{Bagchi84,KLO97,Roughgarden01,Stackelberg52}. Ultimately
one could also apply a VCG mechanism to achieve social efficiency. The
main drawback of these methods is the need for global knowledge of the
system. A different approach, and the focus of our paper, uses
coordination mechanisms~\cite{CKN09}, which only require local
computations.

More formally, a coordination
mechanism~\cite{CKN09,ILMS09,AJM08,C09,DT09} is a set of
\emph{local policies}, one per machine, specifying 
how the jobs assigned to that machine are scheduled.
Here, \emph{local} means that a machine's schedule must be a function only of the jobs
it is assigned, allowing the policy to be implemented in a distributed
fashion. We actually study \emph{strongly} local
policies, meaning that the schedule of any machine $i$ is a function
only of the processing times $\p{i,j}$, weights $w_j$ and IDs of jobs assigned to it. 
It will also be useful (especially when considering lowerbounds) for us to restrict attention to policies that always use the full capacity of a machine, and release jobs immediately upon completion. We call such policies \emph{prompt}.

Several local policies have been studied for machine scheduling problems in
the context of both greedy and local search algorithms~
\cite{IK77,FH79,SC80,DJ81,AAF97,ANR95,BNR02,V02}, as well as
coordination mechanisms~\cite{KP99,CV02,CKN09,ILMS09,AJM08,C09,DT09}.
Previous work mainly considered the makespan social cost as opposed to
the weighted sum of completion times addressed here.

\paragraph{Our Results.}
Employing our new technique, we develop the first constant-factor approximate
coordination mechanisms for the selfish machine scheduling problem for unrelated machines.
We start by studying Smith's Rule \cite{S56}, in which
machines process jobs in increasing order of their processing time to
weight ratio. Here the space that appropriately fits our method is
$L^2$ and a norm distortion inequality is in fact not needed. 
We prove that the approximation factor for this
policy is exactly 4, improving upon a result by Correa and Queyranne~\cite{CQ10}. 
We also show that this is the best possible among all deterministic and non-preemptive strongly local coordination mechanisms, assuming the prompt property.

The constant approximation ratio for the weighted sum of completion times is
in sharp contrast to the known super-constant inapproximability results for coordination 
mechanisms for the makespan function~\cite{AJM08,FS10}  (e.g., an $\Omega(m)$ 
lower bound for the shortest-first coordination mechanism). In fact, it is still open 
whether there is a coordination mechanism with a constant approximation ratio
for the makespan function.

Next, we go beyond the approximation ratio of $4$ using preemptive\footnote{By preemption we mean that the computation of a job is suspended and,
implicitly, resumed later.} and randomized
mechanisms. First, we consider a preemptive policy, generalizing that of D\"urr and
Thang~\cite{DT09}, in which each machine splits its processing
capacity among its assigned jobs in proportion to their weights. We
uncover a close connection of this policy to Smith's Rule, allowing us
to apply a similar proof strategy, but yielding a
significantly improved approximation factor of 2.618.
On the other hand, we prove that with anonymous jobs, no set of deterministic prompt policies,
be they preemptive or not, can achieve a factor better than $2.166$. To break this new barrier we consider a policy 
in which jobs are randomly, but non-uniformly, ordered, based on their processing time to weight ratio.
Under this policy the appropriate space has to be carefully chosen and
uses a rather nonstandard inner product, induced by a Hilbert
matrix, whose $i,j$ entry equals $1/(i+j-1)$. A norm distortion
inequality is then needed to relate the norm in this space to the
original cost of an optimal schedule, leading to yet another
improvement in the approximation factor to $2.134$. Moreover, we show 
a lower bound of $5/3>1.666$ for this policy.

Finally, inspired by our preemptive mechanism, along
with the $\beta$-nice notion of~\cite{AAEMS08}, we design a new
\emph{combinatorial} $(2+\epsilon)$-approximation algorithm for 
optimizing the weighted sum of completion times on unrelated machines. 
This improves on the approximation factor of our mechanisms and complements
the known non-combinatorial  constant-factor approximation algorithms: 
a linear programming based $16/3$-approximation
algorithm~\cite{HSSW97}, then an improvement to $\tfrac32+\epsilon$
again based on linear programming~\cite{SS}, and finally the best currently
known factor, a $\tfrac32$-approximation 
based on a convex quadratic relaxation~\cite{SS99,Skut01}.

We obtain a number of other results, most of which are discussed in the 
appendices. 
For the unit weight case, using Smith's Rule, we obtain a constant upper bound on the price of
anarchy by a reduction from the priority routing model of Farzad et al.~\cite{FOV08},
as shown in Appendix~\ref{appendix:reduction}; however, the  resulting bound is not optimal.
In the unit weight case,  our 
preemptive mechanism simplifies to one, called {\equ} \cite{DT09}, 
in which jobs share the processing capacity of a machine equally. 
Then the approximation ratio is $2.5$, which follows either 
by a careful analysis based on local moves, or using our method with a 
modified Cauchy-Schwartz inequality.
In addition, in the case where the
weighted sum of processing times is negligible compared to the total
cost, our randomized policy has an approximation ratio of $\pi/2$,
which is tight. The latter follows by an interesting norm distortion
inequality obtained by Chung et al.~\cite{Chung1988SelfOrganizing},
for which we provide an alternative shorter proof. Furthermore,
although for the Smith's Rule policy pure equilibria may not exist
\cite{CQ10}, we show that all our preemptive and randomized mechanisms
result in exact potential games. This implies that
best-response dynamics of players converge to pure Nash Equilibria (PNE) and verifies that
PNE always exist. 
While we present our results for pure strategies and pure Nash equilibria,
we observe that  all the results can be stated within the smoothness framework of Roughgarden~\cite{R09}, 
and so all the bounds hold for more general equilibrium concepts including mixed Nash equilibria
and correlated equilibria. We assume that jobs
aim at minimizing their weighted completion time in the case of deterministic
policies, and expected weighted completion time for randomized ones.

It is important to stress that these bounds are on the price of
anarchy (or approximation ratio) of \emph{coordination mechanisms} and
not that of games; thus these results do not follow from seemingly similar
bounds for selfish routing~\cite{AAE05}. The fact that our preemptive
policy performs better than non-preemptive ones is in contrast
to existing results for the makespan social cost function where the
{\equ} policy achieves an approximation ratio of
$\Theta(m)$~\cite{DT09}, no better than {\shf}, which
schedules jobs in increasing order of their processing times (i.e.,
Smith's Rule in the unweighted case). In order to explain this
counter-intuitive result, we show that both our preemptive policy and
our randomized policy penalize each job with an extra
charge beyond its cost under Smith's Rule that is exactly equal to the
externality its scheduling causes.


\paragraph{Other Related Work.}
Scheduling problems have long been studied from a centralized optimization perspective. 
We adopt the standard three filed notation $\alpha| \beta | \gamma$~\cite{GLLK79}. 
The first parameter defines the machine model, the last specifies
the objective function, while the second will not concern us in this
paper.


Minimizing the sum of completion times is polynomial time solvable
even for unrelated machines~\cite{H73,BCS74}. For identical parallel
machines ($\schedprob{P}{\sum c_j}$), the {\shf} policy leads to an optimal
schedule at any pure Nash equilibrium\footnote{In \cite{ILMS09} it is
shown that these equilibria are exactly the solutions generated by the
shortest-first greedy algorithm.}~\cite{CMM67}. On the other hand,
minimizing the weighted sum of completion times is NP-complete even
for identical machines (\idents) \cite{LKB77}. Although the latter
admits a PTAS~\cite{SW00}, the general unrelated case (\unrels) is
APX-hard~\cite{HSW98} and constant factor approximation algorithms
have been proposed \cite{HSSW97,SS,SS99,Skut01}.

Coordination mechanism design was introduced by Christodoulou, Koutsoupias and
Nanavati~\cite{CKN09}. They analyzed the {\lof} policy w.r.t.\ the
makespan for identical machines ({\ident}) and also studied a selfish
routing game. Immorlica et al.~\cite{ILMS09} study four coordination
mechanisms for different machine scheduling problems and survey the
results for these problems.
They further study the speed of convergence to equilibria and the
existence of PNE for the {\shf} and {\lof} policies. Azar, Jain, and
Mirrokni~\cite{AJM08} showed that the {\shf} policy and in fact any strongly local ordering policy (defined in Section~\ref{sec:prelim})
does not achieve an approximation ratio better than $\Omega(m)$. Additionally,
they presented a non-preemptive local policy that achieves an
approximation ratio of $O(\log m)$ and a policy that induces potential
games and gives an approximatition ratio of $O(\log^2 m)$.
Caragiannis~\cite{C09} showed an alternative $O(\log m)$-approximate
coordination
mechanism that minimizes makespan for unrelated machine scheduling and does
lead to potential games. Fleischer and Svitkina~\cite{FS10} show a lower bound
of $\Omega(\log m)$ for all local ordering policies. 
It is still open whether there exists a coordination
mechanism (even preemptive or randomized) achieving
a constant approximation ratio for the makespan objective function.

More recently, D{\"u}rr and Thang proved that the {\equ} policy
results in potential
games, and achieves a PoA of $\Theta(m)$ for \unrel. In the context of
coordination mechanisms, an instance for which a preemptive policy has
an advantage over non-preemptive ones was also shown by
Caragiannis~\cite{C09}, who presented a local preemptive policy with
approximation factor $O(\log m / \! \log\log m)$, beating the lower
bound of $\Omega(\log m)$ for local ordering policies~\cite{FS10}.
Correa and Queyranne~\cite{CQ10} study the problem of minimizing the weighted
sum of completion times, show that Smith's Rule may induce games that
do not have PNE, and that the price of anarchy under this policy is
4 in a more restricted
environment than that considered here.


Coordination mechanisms are related to local search algorithms.
Starting from a solution,
a local search algorithm iteratively moves to a neighbor solution
which improves the global objective.
This is based on a neighborhood relation that is defined on the set of
solutions. The local improvement moves in the local search algorithm
correspond to the best-response moves of
users in the game defined by the coordination mechanism. The speed of
convergence and the approximation factor of local search algorithms
for scheduling problems
have been studied mainly for the makespan objective 
function~\cite{DJ81,EKM03,FH79,IK77,SC80,SV01,V02,AAF97,ANR95}.
Our combinatorial approximation algorithm for the weighted sum of completion
time is the first local search algorithm for $R\vert \vert \sum w_i C_i$ and is
different from the previously studied algorithms for the makespan objective.


\section{Preliminaries}\label{sec:prelim}

Throughout this paper, let $\Jobs$ be a set of $n$ jobs to be scheduled 
on a set $\Macs$ of $m$ machines. 
Let $\p{i,j}$ denote the processing time of job $j\in \Jobs$ on machine $i\in \Macs$ 
and let $w_j$ denote its weight (or importance or impatience). 
Our goal is to minimize the weighted sum of the completion times 
of the jobs, i.e.\ $\sum_{j\in \Jobs} w_j c_j$, where $c_j$ is the completion time
of job $j$. 
An assignment of jobs to machines is represented by a vector $\stv$, where $\st_j$ gives the machine to which job $j$ is assigned.

The main scheduling model we study is \emph{unrelated} machine scheduling (\unrels) in which the $\p{i,j}$'s are arbitrary.
Another model is the \emph{restricted related} machines model in which each machine $i$ has a speed $q_i$ and each job $j$ has a processing requirement $p_j$:  job $j$ can be scheduled only on a subset $T_j$ of the machines, with processing time $\p{i,j}=p_j/q_i$ for $i\in T_j$, and $\p{i,j}=\infty$ otherwise. 
The \emph{restricted identical} machines model is the special case of the restricted related machines model where all machines have the same speed.

A \emph{coordination mechanism} is a set of local policies, one for each machine, that determines how to schedule the jobs assigned to that machine. It thereby defines a game in which there are $n$ agents (jobs) and each agent's strategy set is the set of machines $\Macs$. Given an assignment $\stv$, the disutility of job $j$ is its 
weighted completion time $w_jc_j(\stv)$, as determined by the policy on the machine $x_j$. 
The goal of each job is to choose a strategy (i.e., a machine) that minimizes its disutility. 
%
A strategy profile $\stv$ is a \emph{Nash equilibrium} if no player has an incentive to change strategy.
Our goal is to design coordination mechanisms which give such incentives to the players, that selfish behavior leads to equilibria with low social cost.

A game is a \emph{potential game} if there exists a potential function over strategy profiles such that any player's deviation 
leads to a drop of the potential function if and only if its cost drops. A potential game is \emph{exact} if after each move, the changes to the potential function and to the player's cost are equal. It is easy to see that a potential game always possesses a PNE.


\let\origdescription\description
\renewenvironment{description}{
  \setlength{\leftmargini}{0em}
  \origdescription
  \setlength{\itemindent}{0em}
}
{\endlist}

We define a machine's policy to be \emph{prompt} if the machine uses its full capacity and does not delay the release of any of its completed jobs.
We say that a policy satisfies the \emph{independence of irrelevant alternatives} or \emph{IIA} property if for any
pair of jobs, their relative ordering is independent of what other jobs are assigned to the machine.
This property property appears as an axiom in voting theory, bargaining theory and logic.
Notice that deterministic non-preemptive policies with the IIA property can be described simply by a fixed ordering of all jobs; jobs are scheduled according to this order. 
Thus we call such policies \emph{ordering} policies.

Here and throughout the paper, we use the shorthand notation $\pow{i,j}$ for the ratio $\p{i,j}/w_j$.
The coordination mechanisms we study in this paper use the same local policy on each machine, so henceforth we refer to a coordination mechanism using the name of the policy. 
The main policies we discuss are the following:
\begin{description}
	\item[\smi~\cite{S56}:] Jobs on machine $i$ are scheduled consecutively in increasing order of $\pow{i,j}$. In the unweighted case, this reduces to the \shf\ policy.
	\item[\pro:] Jobs are scheduled in parallel using time-multiplexing. At any moment in time, each uncompleted assigned job receives a fraction of the processor time equal to its weight divided by the total weight of uncompleted jobs on the machine.
		In the unweighted case, this gives the \equ\ policy.
	\item[\rand:] This randomized policy has the property that for any two jobs $j, j'$ assigned to machine $i$, the probability that job $j$ is run before job $j'$ is exactly $\frac{\pow{i,j'}}{\pow{i,j} + \pow{i,j'}}$. Thus larger (w.r.t. $\pow{i,j}$) jobs are more likely to appear later in the ordering. 
		We show how to implement this policy in Section~\ref{sec:rand}.
\end{description}

For any configuration $\stv$, let $w_j c_j^{\alpha}(\stv)$ and $C^{\alpha}(\stv)$ denote the cost for player $j$ and 
the social cost respectively, where $\alpha \in \{SR,PS,SF,ES,R \}$ denotes the 
policy, namely {\smi}, {\pro}, {\shf}, {\equ} and \Rand, respectively. Finally, slightly 
abusing notation, let $\St_i=\{j\in \Jobs~|~\st_j=i\}$ denote the set of jobs that 
have chosen machine $i$ in configuration $\stv$, and define $\St*_i$ analogously
for $\stv*$.

A local policy for machine $i$ uses only the information about the jobs on the 
same machine $i$, but it can look at all the parameters of these jobs, including 
their processing times on other machines. By contrast, a \emph{strongly local} 
policy may depend only on the processing time that these jobs have on this machine $i$. 

In order to quantify the inefficiency caused by the lack of coordination, we use the notion of \emph{price of anarchy}~\cite{KP99}, that is, the ratio between the social cost value of the worst Nash equilibrium and that of the social optimum. 
To be more precise, we are interested in upper bounds for the PoA of coordination mechanisms rather than the PoA
of specific games. Applying~\cite{CKN09} to the current context, the PoA of a coordination mechanism is defined to be the maximum ratio, taken 
over all the games $G$ that the mechanism may induce, of the social cost of a Nash equilibrium of $G$ divided by the optimum social cost achievable for the scheduling problem underlying $G$. 


\section{Deterministic Non-Preemptive Coordination Mechanisms}

It is known that given an assignment of jobs to machines, in order to minimize the weighted 
sum of completion times, {\smi} is optimal~\cite{S56}. It is therefore only natural to consider
this policy as a good first candidate. Our first theorem shows
that using this rule will result in Nash equilibria with social cost at most a constant-factor of 4 away
from the optimum.

Our analysis uses the map $\sig: \Macs^\Jobs \rightarrow L_2([0, \infty))^\Macs$, which maps a configuration to a vector of functions as follows. If $\vec{f} = \sig(\stv)$, then
\[ f_i(y) = \sum_{\mathrlap{j \in \St_i: \pow{i,j} \geq y}}\, w_j ~~~~~~~~\mbox{(recall that $\pow{i,j} = \p{i,j}/w_j$).}\]

We let $\langle g, h \rangle := \int_0^\infty g(y)h(y)dy$ denote the usual inner product on $L_2$, and in addition define $\langle \vec{f}, \vec{g} \rangle := \sum_{i \in \Macs} \langle f_i, g_i \rangle$. In both cases, $\| \cdot \|$ refers to the induced norm.
We also define
\[ \niggly{x} = \sum_{j\in \Jobs} w_j p_{x_j j}. \]
We then have
\begin{lemma}\label{lem:smithcost} 
  For any configuration $\stv$, 
$\Csmith(\stv) = \tfrac12 \langle \sig(\stv), \sig(\stv)\rangle + \tfrac12\niggly{x}$. 
\end{lemma}
\begin{proof}
Let $\vec{f} = \sig(\stv)$. 
  We have
  \begin{align*} 
	  \langle \sig(\stv), \sig(\stv)\rangle &= \sum_{i \in I} \int_0^\infty f_i(y)^2 dy\\
  &= \sum_{i \in I} \sum_{j \in \St_i}\sum_{k \in \St_i} w_jw_k \int_0^{\infty} \ind_{\pow{i,j} \geq y} \ind_{\pow{i,k} \geq y} dy \\
  &= \sum_{i \in I} \sum_{j \in \St_i}\sum_{k \in \St_i} w_jw_k\min\{\pow{i,j}, \pow{i,k}\}\\
  &= \sum_{i \in I} \sum_{j \in \St_i}w_j\Bigl(\,\,2\!\!\!\sum_{\substack{k \in \St_i\\ \pow{i,k} \leq \pow{i,j}}} \!\!\p{i,k}
  \,\,-\, \p{i,j} \Bigr)\\
  &= 2\Csmith(\stv) - \niggly{x}.
  \end{align*}
The result follows.
\end{proof}

\begin{theorem}\label{thm:smi}
	The 
 \poa\ of 
{\smi} for unrelated machines \textup{(\unrels)} is at most $4$.
\end{theorem}
\begin{proof} 
    Let $\stv$ and $\stv*$ be two assignments, with $\stv$ a Nash equilibrium,
    and write $\vec{f}=\sig(\stv)$, $\vec{f^*} = \sig(\stv*)$. 
    We assume for simplicity that all jobs have distinct ratios (of processing time to weight). 
   (This assumption is just for simplicity; alternatively, we could introduce a tie breaking rule.)

    We first calculate a job $j$'s completion time according to $\stv$, and use the Nash condition:
  \[ \csmith_j = \sum_{\substack{k: x_k = x_j\\ \pow{x_k,k} < \pow{x_j,j}}}\!\!\! \p{x_k,k} \,\,+\, \p{x_j,j} \leq \sum_{\substack{k: x_k = x^*_j\\ \pow{x_k,k} < \pow{x^*_j,j}}} \!\!\!\p{x_k,k} \,\,+\, \p{x^*_j,j}.
    \]
 
    \begin{align}
  \mbox{So}~~~~~~~~~~~~~~~     
    \Csmith(\stv) = \sum_j w_j \csmith_j &\leq \sum_{i \in \Macs} \sum_{j \in \St*_i} \Bigl( \sum_{\substack{k \in \St_i\\\pow{i,k} < \pow{i,j}}}w_jw_k\frac{p_{ik}}{w_k} + \p{i,j}w_j\Bigr) \notag\\
		&\leq \sum_{i \in \Macs} \sum_{j \in \St*_i} \sum_{k \in \St_i}w_jw_k\min\{\pow{i,k}, \pow{i,j}\} \,+\, \sum_{i \in \Macs} \sum_{j \in \St*_i} \p{i,j}w_j \notag \\ 
  & = \sum_{i \in \Macs} \sum_{j \in \St*_i} \sum_{k \in \St_i}w_jw_k \int_0^{\infty} \ind_{\pow{i,j} \geq y} \ind_{\pow{i,k} \geq y} dy \,+\, \niggly{x^*} ~~~~~~~~~~~ \notag\\
    &= \ip{\vec{f^*}, \vec{f}} \,+\, \niggly{x^*}.\notag
    \end{align}

  Now applying Cauchy-Schwartz, followed by the inequality $ab \leq a^2 + b^2/4$ for $a,b \geq 0$, we obtain
    \begin{align*}
    \Csmith(\stv) &\leq \norm{\vec{f}}\norm{\vec{f^*}} \,+\, \niggly{x^*}\\
    &\leq \norm{\vec{f^*}}^2 + \tfrac14\norm{\vec{f}}^2 + \niggly{x^*}\\
    &\leq 2\Csmith(\stv*) + \tfrac12 \Csmith(\stv) \qquad \text{by Lemma~\ref{lem:smithcost}}.
  \end{align*}
  Hence $\Csmith(\stv) \leq 4\Csmith(\stv*)$.
\end{proof}

The following result, proved in Appendix~\ref{appendix:lowerbound}, shows that (assuming promptness) no deterministic non-preemptive strongly local mechanism can do better than {\smi}.
This also implies that the bound of Theorem~\ref{thm:smi} is tight.
\begin{theorem}\label{thm:lowb}
The pure PoA of any strongly local deterministic non-preemptive prompt coordination mechanism is at least 4. 
This is true even for the case of restricted identical machines (\bipars) with unweighted jobs.
\end{theorem}

\section{Improvements with Preemption and Randomization}\label{sec:better}

\subsection{Preemptive Coordination Mechanism}

In this section, we study the power of preemption and present {\pro}, a preemptive mechanism that is strictly better w.r.t.\ the PoA than any deterministic non-preemptive strongly local policy. These results create a clear dichotomy between such policies and {\pro}.
This may seem counter-intuitive at first, since, given an assignment of jobs to machines, using {\pro} instead of {\smi} only increases the social cost\footnote{Note that this is not the case for the makespan social cost function.} and doesn't decrease the cost of any player.

A better understanding of this result can be obtained by observing that in our context, preemptive policies can be thought of (and also implemented) as non-preemptive but also non-prompt policies.
Jobs are run in an appropriate order, but possibly delayed past their completion time.
(Notice however that such an implementation would technically disallow anonymous jobs, i.e., jobs that do not have IDs.)
From this perspective, {\pro} can be implemented by using {\smi} to determine the processing order, but then
holding each job back after it is completed by an amount equal to the total delay it 
causes to other jobs {\smi} schedules after it.
This fact can be seen explicitly in the first equation of the upcoming Lemma~\ref{lem:proc}.
In this way, the interests of a player are aligned with those of the group by having it ``internalize its externalities'', leading not only to better allocations but also to a better social cost, despite the extra charges. 
Additional advantages of this coordination mechanism are that, unlike {\smi}, it can handle anonymous jobs, and the games it induces always possess PNE.

\begin{lemma}\label{lem:proc}
	Given an assignment $\stv$, the weighted completion time of a job $j$ on some machine $i$ using {\pro}
	(whether currently assigned there or not) is 
\begin{align}
	w_j\cps_j=&\sum_{\mathrlap{\substack{k \in \St_i \setminus \{j\} \\ \pow{i,k} \leq \pow{i,j}}}}{w_j\p{i,k}}\,+\, \sum_{\mathrlap{\substack{k \in \St_i\\ \pow{i,k} > \pow{i,j}}}}{w_k\p{i,j}} \,\,+\,\, w_j\p{i,j} \notag \\
	=& \sum_{\mathrlap{k \in \St_i \setminus \{j\}}} w_jw_k \min\{\pow{i,k}, \pow{i,j}\}\,\,+\,\, w_j\p{i,j}.\label{for:cost}
\end{align}
\end{lemma}
\begin{proof}

We notice that the completion time of job $j$ is only affected by the amount of ``work'' that the processor has completed by that time and not by the way this processing time has been shared among the jobs. For job $j$ and any job $k$ such that $\pow{i,k} \leq \pow{i,j}$, we know that their whole processing demands, $\p{i,j}$ and $\p{i,k}$ respectively, have been served. On the other hand, while job $j$ is not complete, for each $w_j$ units of processing time it receives, any job $k$ with $\pow{i,k} > \pow{i,j}$ receives $w_k$ units. Thus, when job $j$ is completed, the processing time spent on any such job $k$ will be exactly $\frac{p_{ij} w_k}{w_j}$. Adding all these processing times and multiplying by player $j$'s weight, $w_j$ gives the lemma.
\end{proof}

\begin{theorem}\label{thm:pro}
	The \poa\ of {\pro} for unrelated machines \textup{(\unrels)} is at most $\phi+1=\frac{3+\sqrt{5}}{2}\approx 2.618$.
Moreover, this bound is tight even for the restricted related machines model.
\end{theorem}
\begin{proof}
	By Lemma~\ref{lem:proc}, we see that for any assignment $\stv$, $\Cps(\stv) = \norm{\sig(\stv)}^2$; 
note the factor two difference compared to the first term for \smi.
Moreover, \eqref{for:cost} is upper bounded by
\[ \sum_{k \in \St_i} w_j w_k\min\{\pow{i,j}, \pow{i,k}\} \,+\, w_j \p{i,j}, \]
and so the Nash condition implies that for any equilibrium $\stv$, and any other assignment $\stv*$,
\begin{align*}
 \Cps(\stv) &\leq \sum_j\left(\sum_{k: \st_k = \st_j^*} w_jw_k \min\{\pow{\st_j^*,j}, \pow{\st_j^*,k} \} \,+\, \p{\st_j^*,j}\right)\\
&= \ip{\sig(\stv), \sig(\stv*)} + \niggly{\stv*}.
\end{align*}
This is identical to the equation we used (after an additional inequality) in the case of Smith's Rule.

Let $\stv$ be a Nash assignment, $\stv*$ an optimal assignment w.r.t.\ Smith's Rule, and again define $\vec{f} = \sig(\stv)$, $\vec{f^*} = \sig(\stv*)$.
Following the same method of analysis as for Smith's Rule, we obtain
\begin{align*}
    \Cps(\stv) &\leq \norm{\vec{f}}\norm{\vec{f^*}} \,+\, \niggly{\stv*}\\
    &\leq \alpha\norm{\vec{f^*}}^2 + \tfrac{1}{4\alpha}\norm{\vec{f}}^2 + \niggly{\stv*}\\
    &\leq 2\alpha\Csmith(\stv*) + \tfrac{1}{4\alpha}\Cps(\stv) + (1-\alpha)\niggly{\stv*}\\
	&\leq (1+\alpha)\Csmith(\stv*) + \tfrac{1}{4\alpha}\Cps(\stv),
\end{align*}
using that $\niggly{\stv*} \leq \Csmith(\stv*)$.
Setting $\alpha=(1+\sqrt{5})/4$ yields ${\Cps(\stv)}/{\Csmith(\stv*)} \leq \frac{3+\sqrt{5}}{2}$.

The tightness of this bound follows from a construction in \cite{CFKKM06}, where in fact they show that even if $\Cps$ is used for the cost of \OPT, i.e., we consider the ratio $\Cps(\stv)/\Cps(\stv*)$, this can be arbitrarily close to $1 + \phi$.
\end{proof}

In the case of equal weights, we obtain the following slightly improved bound. 
This result can be proven in our framework but using a variation of the 
Cauchy-Schwartz inequality derived from Lemma~\ref{lem:ineq} below. However, we present here a different proof approach of independent interest.
\begin{theorem}\label{thm:equ}
	The \poa\ of {\equ} for unrelated machines \textup{($R|~|\sum c_j$)} is at most $2.5$. This bound is tight even for the restricted related machines model.
\end{theorem}
\begin{proof}
We begin by proving the following lemma, which gives a tighter version of an inequality initially used by Christodoulou and Koutsoupias~\cite{CK05}:
\begin{lemma}\label{lem:ineq}
For every pair of non-negative integers $k$ and $k^*$,
\begin{equation*}
k^*(k+1)\leq \frac{1}{3}k^2 + \frac{5}{3}\frac{k^*(k^*+1)}{2}.
\end{equation*}
\end{lemma}
\begin{proof}
After some algebra, this translates to showing that for all non-negative integers $k$ and $k^*$,
\begin{equation*}
5k^{*2}+2k^2-6k^*k-k^* \geq 0.
\end{equation*}
We start by taking the partial derivative of the LHS w.r.t.\ $k$, i.e.\ $4k-6k^*$, from which we infer that for any given value of $k^*$, the LHS is minimized when $k=\frac{3}{2}k^*$. On substituting this into our inequality, we obtain:
\begin{equation*}
5k^{*2}+2(\frac{3}{2}k^*)^2-6k^*\frac{3}{2}k^*-k^* \geq 0 \Rightarrow k^{*2}\geq 2k^*,
\end{equation*}
which is true for $k^*=0$ and $k^*\geq 2$. For $k^*=1$ our inequality becomes $k^2-3k+2\geq 0$ which is true for all non-negative integers $k$.
\end{proof}
%

Now, using this lemma, we show that for any machine $i$:
\begin{equation*}
\sum_{j\in \St*_i}{c_j^{ES}(\stv_{-j},\st*_j)} \leq \frac{1}{3} \sum_{j\in
\St_i}{c_j^{ES}(\stv)}+\frac{5}{3}\sum_{j\in \St*_i}{c_j^{SF}(\stv*)}.
\end{equation*}
In order to show this, we first prove that this inequality only becomes tighter if for any
two jobs $j,j' \in \St_i \cup \St*_i$, their processing times on $i$ are equal, i.e.\
$\p{i,j}=\p{i,j'}$. Assume that not all processing times are equal and let
$\mathit{Max}_i=\{j\in \St_i\cup \St*_i~|~\forall j'\in \St_i \cup \St*_i,~\p{i,j}\geq
\p{i,j'}\}$ be the set of jobs of maximum processing time among the two sets. Also, let
$k^*=| \mathit{Max}_i \cap \St*_i|$ and $k=| \mathit{Max}_i \cap \St_i|$ be the number of
maximum size jobs in sets $\St*_i$ and $\St_i$ respectively.

For all the jobs $j\in \mathit{Max}_i$, we decrease $\p{i,j}$ by the minimum positive value
$\Delta$ such that the cardinality of $\mathit{Max_i}$ increases. After a change of this sort, the LHS drops by $(k^*(k+1))\Delta$ while the RHS drops by
$(\frac{1}{3}k^2+\frac{5}{3}\frac{k^*(k^*+1)}{2})\Delta$. Given Lemma \ref{lem:ineq} above, we conclude that the drop of the LHS is always less than or equal to the drop of the RHS. Using the same
inequality again, we conclude that for unit jobs on machine $i$ the inequality is always true; summing up over all $i\in \Macs$ yields:
\begin{equation*}
\sum_{j\in \Jobs}{c_j^{ES}(\stv_{-j},\st*_j)} \leq \frac{1}{3} \sum_{j\in \Jobs}{c_j^{ES}(\stv)}+\frac{5}{3}\sum_{j\in \Jobs}{c_j^{SF}(\stv*)}.
\end{equation*}
This gives a \poa\ bound of $2.5$.

The tightness of the bound follows from Theorem 3 of \cite{CFKKM06}. The authors present a load balancing game lower bound, which is equivalent to assuming that all jobs have unit size and the machines are using {\equ}; thus the same proof yields a (pure) PoA lower bound for restricted related machines and unweighted jobs.
\end{proof}

%
%

On the negative side, we have the following (the proof of which can be found in Appendix~\ref{appendix:lowerbound})
\begin{proposition}\label{prop:lbdet}
	When jobs are anonymous, the worst-case PoA of any deterministic prompt coordination mechanism is at least $13/6$.
\end{proposition}

\subsection{Randomized Coordination Mechanism}\label{sec:rand}

In this section we examine the power of randomization and present \rand,
which outperforms any prompt deterministic strongly local policy.
Under \rand, for any pair of jobs on the same machine, the externalities
they  cause each other are shared equally in expectation. This is achieved
with the following property: if two jobs $j$ and $j'$ are assigned to machine $i$, then 
\begin{equation}\label{eq:prand}
 \Pr\{j \text{ precedes } j' \text{ in the ordering}\} = \frac{\rho_{ij'}}{\rho_{ij} + \rho_{ij'}}.
\end{equation}
Recall $\pow{i,j} = \p{i,j}/w_j$.
A distribution over orderings with this property can be constructed as follows. Starting from the set of jobs $\St_i$ assigned to machine $i\in I$, select job $j\in \St_i$ with probability $\pow{i,j}/\sum_{k\in \St_i}\pow{i,k}$, and schedule $j$ at the end. Then remove $j$ from the list of jobs, and repeat this process.
%
Note that this policy is different from a simple randomized policy that orders jobs uniformly at random. 
In fact, this simpler policy is known to give an $\Omega(m)$ PoA bound for the makespan function~\cite{ILMS09}, and the same family of examples developed in~\cite{ILMS09} gives an $\Omega(m)$ lower bound for this policy in our setting. Nevertheless, we will prove the following bounds:
\begin{theorem}\label{thm:randupper}
  The price of anarchy when using the \Rand\ policy is at most $32/15 = 2.133\cdots$. 
Moreover, if the sum of the processing times of the jobs is negligible compared to the social cost of the optimal solution, this bound improves to $\pi/2$, which is tight.
\end{theorem}

The high level approach for obtaining the upper bound is in exactly the same spirit as the previous section: find an appropriate mapping $\sig$ from an assignment into a convenient inner product space. 

For simplicity, we assume in this section that the processing times have been scaled such that the ratios $\pow{i,j}$ are all integral. 
This assumption is inessential and easily removed.
We also take $\Sz$ large enough so that, except for infinite processing times, $\pow{i,j}\leq \kappa$ for all $i \in I, j \in J$. 


\paragraph*{An inner product space.}
The map $\sig$ we use gives the \emph{signature} for each machine: in the unweighted case, this simply describes how many jobs of each size are assigned to the machine.
\begin{definition}
Given an assignment $\stv$, its \emph{signature} $\sig(\stv) \in \Rplus^{m \times \Sz}$ is a vector indexed by a machine $i$ and a processing time over weight ratio $r$; we denote this component by $\sig(\stv)^i_r$. Its value is then defined as
\[ \sig(\stv)^i_r := \sum_{\substack{j \in \St_i\\\pow{i,j} = r}} w_j. \]
We also let $\sig(\stv)^i$ denote the vector $(\sig(\stv)^i_0, \sig(\stv)^i_1, \ldots ,\sig(\stv)^i_{\Sz}) $.
\end{definition}

Let $M$ be the $\Sz \times \Sz$ matrix given by 
\[ M_{rs} = \frac{rs}{r+s}. \]
\begin{lemma}\label{lem:randcost}
  Let $\stv$ be some assignment, and let $\vec{u} = \sig(\stv)$.
If job $j$ is assigned to machine $i$, its expected completion time is given by 
\[ \crand_j = (M \vec{u}^i)_{\pow{i,j}} + \tfrac12\p{i,j}. \]
If $j$ is not assigned to $i$, then its expected completion time upon switching to $i$ would be
\[ \crand_j = (M \vec{u}^i)_{\pow{i,j}} + \p{i,j}. \]
\end{lemma}
\begin{proof}
We consider case (i); (ii) is similar. So $x_j = i$.
The expected completion time of job $j$ on machine $i$ is
\begin{align*}
	\crand_j &= \sum_{k\in \St_i \setminus \{j\}} \p{i,k}\Pr\{\text{job $k$ ahead of job $j$}\} \,+\, \p{i,j}\\
	&= \sum_{k \in \St_i \setminus \{j\}} \p{i,k}\frac{\pow{i,j}}{\pow{i,j} + \pow{i,k}} \,+\, \p{i,j}\\
&= \sum_{k \in \St_i} \p{i,k}\frac{\pow{i,j}}{\pow{i,j} + \pow{i,k}} \,+\, \tfrac12\p{i,j}.
\intertext{We can rewrite this in terms of the signature as}
\crand_j &= \sum_s  u_s^i M_{\pow{i,j}s} \,+\, \tfrac12\p{i,j} = 
(M \vec{u}^i)_{\pow{i,j}} + \tfrac12\p{i,j}.\qedhere
\end{align*}
\end{proof}

A crucial observation is the following:
\begin{lemma}
The matrix $M$ is positive definite.
\end{lemma}
\begin{proof}
  Let $D$ be the diagonal matrix with $D_{rr} = r$. Then we have $M = DHD$, 
where the $\Sz \times \Sz$ matrix $H$ is given by $H_{rs} = \frac{1}{r+s}$. This is a submatrix of the infinite Hilbert matrix
$\left(\frac{1}{r+s-1}\right)_{r,s \in \N}$. The Hilbert matrix has the property that it is \emph{totally positive}~\cite{Choi83}, meaning that the determinant of any submatrix is positive. It follows immediately that $H$ is positive definite, and hence so is $M$.
\end{proof}

Thus we may define an inner product by
\begin{equation}\label{eq:randip}
	\ip{\vec{u}, \vec{v}} := \sum_{i \in I} (\vec{u}^i)^T M \vec{v}^i, 
\end{equation}
with an associated norm $\norm{\cdot}$. 
In addition, the total cost $\sum_j w_j \crand_j(\stv)$ of an assignment $\stv$ may be written in the convenient form
\begin{equation}\label{eq:costrand}
  \Crand(\stv) = \norm{\sig(\stv)}^2 + \tfrac12\niggly{x}.
\end{equation}

\paragraph*{Competitiveness of \Rand\ on a single machine.}
How well \Rand\ performs on a \emph{single} machine, compared to the optimal \smi, turns out to play an important role. 
So suppose we have $n$ jobs with size $p_j$ and weight $w_j$, for $j \leq n$. 
The signature $\vec{u}$ is given by just $u_r = \sum_{j: p_j/w_j = r} w_j$. 
Notice that the weighted sum of completion times according to \smi\ and \rand\ respectively are
\[\textstyle\vec{u}^T S \vec{u} + \frac12 \sum_j w_jp_j \qquad \text{and} \qquad \vec{u}^T M \vec{u} +  \frac12 \sum_j   w_jp_j, \]
where $S_{rs} = \tfrac12\min(r,s)$. 
The extra $\sum_j w_jp_j$ terms only help, and in fact turn out to be negligible in the worst case example; ignoring them, the goal is to determine
$\max_{\vec{u} \geq \vec{0}} \frac{\vec{u}^TM\vec{u}}{\vec{u}^T S \vec{u}}$.
So the question is closely related to the worst-case distortion between two norms.

Interestingly, it turns out that this problem has been considered, and solved, in a different context. 
In~\cite{Chung1988SelfOrganizing}, Chung, Hajela and Seymour consider the problem of \emph{self-organizing sequential search}. 
In order to prove a tight bound on the performance of the ``move-to-front'' heuristic compared to the optimal ordering, they show:
%
%
\begin{theorem}[{\cite{Chung1988SelfOrganizing}}]\label{thm:chung}
  For any sequence $u_1, u_2, \ldots, u_k$ with $u_r > 0$ for all $r$,
\begin{equation}
  \sum_{r,s} u_ru_s \frac{rs}{r + s} < \frac{\pi}{4}\sum_{r,s} u_ru_s\min\{r, s\}.
\end{equation}
\end{theorem}
Moreover, this is tight~\cite{Gonnet1981Exegesis}  (take $p_j = 1/j^2$, $w_j=1$, and let $n\rightarrow \infty$).
We also present a quite different proof of the theorem in Appendix~\ref{appendix:randsingle}.
All in all, we find that $\pi/2$ is a tight upper bound on the competitiveness of \Rand\ on a single machine.
The following lemma (which may also be cast as a norm distortion question), is much more easily demonstrated:
\begin{lemma}\label{lem:crude}
For any assignment $\stv$, we have $\Crand(\stv) \leq 2\Csmith(\stv) - \niggly{\stv}$.
\end{lemma}
\begin{proof}
Consider a particular machine $i$. We have
\begin{align*}
 \sum_{j,k \in \St_i} w_jw_k\randterm{i,j,k} &= \sum_{j \neq k \in \St_i} w_jw_k\randterm{i,j,k} + \tfrac12\sum_{j \in \St_i} w_jp_{ij}\\
&\leq \sum_{j \neq k \in \St_i} w_jw_k \smithterm{i,j,k} + \tfrac12\sum_{j \in \St_i}w_jp_{ij}\\
&= \sum_{j, k \in \St_i} w_jw_k \smithterm{i,j,k} - \tfrac12\sum_{j \in \St_i}w_jp_{ij}.
\end{align*}
Summing over all machines gives
\[
\Crand(\stv) - \tfrac12 \niggly{\stv} \leq  2 \bigl(\Csmith(\stv) - \tfrac12\niggly{\stv}\bigr) - \tfrac12  \niggly{\stv}
\]
from which the bound is immediate.
\end{proof}

\paragraph*{The upper bound.}
We are now ready to prove the main theorem of this section. 

\begin{proof}[Proof of Theorem~\ref{thm:randupper}]
Let $\stv$ be the assignment at a Nash equilibrium, and $\stv*$ the assignment of the optimal solution, 
and let $\vec{u} = \sig(\stv)$ and $\vec{u^*} = \sig(\stv*)$.

From the Nash condition and Lemma~\ref{lem:randcost}, we obtain
\begin{align*}
    \Crand(\stv) &\leq \sum_{j \in \Jobs} w_j \crand_j(\stv_{-j}, \st*_j)\\
    &\leq \sum_{i \in \Macs}\sum_{j \in \St*_i}w_jM(\vec{u}^i)_{\rho_{ij}} \,+\, \niggly{\stv*}\\
    &= \sum_{i \in \Macs} (\vec{u^*}^i)^T M \vec{u}^i \,+\, \niggly{\stv*}\\
    &= \ip{\vec{u^*}, \vec{u}} + \niggly{\stv*}.
\end{align*}
Applying Cauchy-Schwartz
\begin{align}
	\Crand(\stv) &\leq \norm{\vec{u^*}}\norm{\vec{u}} + \niggly{\stv*} \label{eq:cs}\\
&\leq \tfrac23\norm{\vec{u^*}}^2 + \tfrac{3}{8}\norm{\vec{u}}^2 + \niggly{\stv*},\notag
\end{align}
Now recalling the definition of $\sig$ and applying Lemma~\ref{lem:crude}, we obtain
\begin{align*}
    \Crand(\stv) &\leq \tfrac23(\Crand(\stv*) - \tfrac12\niggly{\stv*}) + \tfrac{3}{8}(\Crand(\stv)-\tfrac12\niggly{\stv}) + \niggly{\stv*}\\
	&\leq \tfrac23(2\Csmith(\stv*) - \tfrac32\niggly{\stv*}) + \tfrac38(\Crand(\stv)-\tfrac12\niggly{\stv}) + \niggly{\stv*}\\
    &\leq \tfrac43\Csmith(\stv*) + \tfrac38\Crand(\stv).
\end{align*}
This gives a PoA of $32/15$. 

\medskip
In the case where $\niggly{\stv*}$ is negligible, we continue from~\eqref{eq:cs}:
\begin{align*} 
	\Crand(\stv) &\leq \norm{\vec{u^*}}\norm{\vec{u}}\\
	&\leq \tfrac12\norm{\vec{u^*}}^2 + \tfrac12\norm{\vec{u}}^2\\
	&\leq \tfrac{\pi}{4}\Csmith(\stv*) + \tfrac12\Crand(\stv), 
\end{align*}
by Theorem~\ref{thm:chung} and Equation~\ref{eq:costrand}. Thus $\Crand(\stv)/\Csmith(\stv*) \leq \pi/2$.
\end{proof}

As noted in Appendix~\ref{appendix:lowerbound}, a slight modification of the construction used to prove Proposition~\ref{prop:lbdet} can be used to show that the worst-case PoA of {\rand} is at least $5/3$.

\section{Existence of PNE and Algorithm}

\paragraph*{Existence of PNE.}

Under {\smi} it may happen that no pure Nash equilibrium exists 
\cite{CQ10}. 
Here we show that \pro\ and \Rand\ both induce exact potential games, which hence always have PNE. 
For the case of \pro, this generalizes \cite[Theorem 3]{DT09}, which addresses \equ.

\begin{theorem}\label{thm:potential}
The {\pro} mechanism induces exact potential games, with potential
\begin{equation}\label{eq:pspot}
	\PotPS(\stv)=\tfrac{1}{2}\Cps(\stv) + \tfrac12\niggly{\stv}.
\end{equation}
Likewise, the \Rand\ mechanism yields exact potential games with potential
\begin{equation}\label{eq:randpot}
	\PotR(\stv)=\tfrac12\Crand(\stv) + \tfrac12\niggly{\stv}. 
\end{equation}
\end{theorem}
\begin{proof}
	We give the proof for {\pro}; the proofs for \Rand\ and {\apr} are similar.


Consider an assignment $\stv$ and a job $j \in \Jobs$, and let $i$ be the machine to which $j$ is assigned.
Define $\stv'$ as the assignment differing from $\stv$ only in that job $j$ moves to some machine $i' \neq i$.

We may write the change in the potential function as
\begin{align}
	\PotPS(\stv') - \PotPS(\stv) &=  \sum_{k \in \Jobs} \PotD_k \,+\, \tfrac12w_j(\p{i',j}-\p{i,j}),\label{eq:potsplit}\\
	\intertext{where}
	 \PotD_k &= \tfrac12w_k\left(\cps_k(\stv') - \cps_k(\stv)\right).\notag
\end{align}
Consider a job $k \neq j$ on machine $i$. Since only job $j$ left the machine, we have from Lemma~\ref{lem:proc} that
\[ \cps_{k}(\stv') - \cps_{k}(\stv) = -w_j\min\{\pow{i,j}, \pow{i,k}\}. \]
Thus
\begin{align*}
 \sum_{k \in \St_i \setminus \{j\}} \PotD_k &= -\tfrac12w_j\sum_{k \in \St_i \setminus \{j\}} w_k\min\{\pow{i,j},\pow{i,k}\}\\
&= -\tfrac12w_j(\cps_j(\stv) + \p{i,j}). 
\end{align*}
Similarly, considering jobs on $i'$ yields
\begin{align*}
 \sum_{k \in \St_{i'}} \PotD_k &= \tfrac12w_j\sum_{k \in \St_{i'}} w_k\min\{\pow{i',j},\pow{i',k}\}\\
&= \tfrac12w_j(\cps_j(\stv') - \p{i',j}). 
\end{align*}
All other jobs are unaffected by the change, and so do not contribute to~\eqref{eq:potsplit}. Summing all terms (including $\PotD_j$), we obtain
\[ \PotPS(\stv') - \PotPS(\stv) = w_j(\cps_j(\stv') - \cps_j(\stv)), \]
exactly the change in the cost 
of job $j$.
\end{proof}

\paragraph*{A combinatorial approximation algorithm.}
Finally we define {\apr}, a deterministic strongly local policy that we will use in order to design a
combinatorial constant factor approximation algorithm for the underlying optimization problem. The completion
time of a job in this policy is exactly its completion time if {\pro} were being used plus its processing time, i.e. $c_j^A(\stv)=c_j^{PS}(\stv)+\p{\st_j,j}$.
%

Following the proof of Theorem~\ref{thm:potential} we can show that $\PotA(\stv)=\tfrac{1}{2}\Capr(\stv) + \niggly{\stv}$ is a potential function
for the games induced by this mechanism, and following the proof of Theorem~\ref{thm:pro}, the PoA of the mechanism is at most 4. The advantage of this mechanism is that $\Capr(\stv)=2\Csmith(\stv)$ for any configuration $\stv$ and therefore, despite the larger PoA bound, computing an
equilibrium allocation for the induced game yields a scheduling algorithm with approximation ratio 2
(because the scheduling algorithm, given the allocation, applies {\smi} and not {\apr}).

Computing such an allocation might be hard in general, but we show that imitating a natural best response dynamics
gives rise to a simple polynomial time local search ($2+\epsilon$)-approximation algorithm.
At each iteration, the scheduling algorithm reassigns the job which thereby obtains the largest 
possible improvement in the {\apr} costing (a best response move).
In other words, we use {\apr} in order to find a good allocation and then switch to {\smi}.

%
In order to bound the running time of our local-search algorithm we will use the 
$\beta$-nice concept of \cite{AAEMS08}.
Given some configuration $\stv$, let 
\[ \Delta(\stv)=\sum_{j}{\left(c_j(\stv)-c_j(\stv_{-j},\st_j')\right)},\] 
where $\st_j'$ is the best response to $\stv_{-j}$ for player $j$. Awerbuch et al.~\cite{AAEMS08} define an 
exact potential game with potential function $\Phi$ and social cost function $C$ to be $\beta$-nice if and 
only if, for any configuration $\stv$, both $\Phi(\stv)\leq C(\stv)$ and $C(\stv)\leq \beta \OPT+2\Delta(\stv)$
hold\footnote{In their definition, unlike ours, \OPT\ denotes the optimum social cost w.r.t.\ the game's cost
functions.}. Among other dynamics, they consider what they call \emph{basic dynamics}, where in each step,
among all players that can uniquely deviate and improve their cost by some factor $\alpha$, we choose the
one with the largest absolute improvement, and allow that player to move. They subsequently show the following lemma,
where $\stv*$ is the configuration that minimizes the potential function.

\begin{lemma}[{\cite{AAEMS08}}]\label{lem:basic}
Let $\frac{1}{8}>\epsilon>\alpha$. Consider an exact potential game that satisfies 
the $\beta$-nice property and any initial state $\stv^0$. Then basic dynamics generates a profile 
$\stv$ with $C(\stv)\leq \beta(1+O(\epsilon))\OPT$ in at most $O\left(\frac{n}{\epsilon}\log{\left(\frac{\Phi(\stv^0)}{\Phi(\stv*)}\right)}\right)$ steps.
\end{lemma}

We define a \emph{coordination mechanism} to be $\beta$-nice if all the games that it induces are 
$\beta$-nice with \OPT\ being the optimum social cost of the underlying machine scheduling problem, independent 
of the coordination mechanism. Our next lemma shows that {\apr} satisfies these conditions.

\begin{lemma}\label{lem:beta}
The {\apr} coordination mechanism is $\beta$-nice with $\beta=4$.
\end{lemma}

\begin{proof}
It is easy to see that the potential function $\PotA(\stv)=\tfrac{1}{2}\Capr(\stv) + \niggly{\stv}$ satisfies $\PotA(\stv)\leq \Capr(\stv)$ for all configurations $\stv$, since $\niggly{\stv} \leq \tfrac12 \Capr(\stv)$. Therefore, what we need to show is that:
\begin{equation*}
\Capr(\stv)\leq \beta \Csmith(\stv*)+2\Delta(\stv),
\end{equation*}
where 
\begin{equation*}
\Delta(\stv)=\sum_{j\in \Jobs}{\left(w_j\capr_j(\stv)-w_j\capr_j(\stv_{-j},\st_j')\right)},
\end{equation*}
and $\st_j'$ is the best response for player $j$ in configuration $\stv$. We note that since $\capr_j(\stv_{-j},\st_j')\leq \capr_j(\stv_{-j},\st*_j)$, then:
\begin{equation*}
\Capr(\stv)-\sum_{j\in \Jobs}{w_j\capr_j(\stv_{-j},\st*_j)}\leq \Delta(\stv).
\end{equation*}
Now following the same approach as in the proof of Theorem~\ref{thm:pro}, we easily obtain that the PoA is at most $4$.
More specifically, we can obtain the inequality
\[ \sum_{j \in \Jobs} w_j\capr_j(\stv_{-j}, \st*_j) \leq \tfrac14\Capr(\stv) + 3\Csmith(\stv*). \]
Summing these two inequalities and simplifying we obtain
\begin{equation*}
\Capr(\stv)\leq 4 \Csmith(\stv*) + \tfrac{4}{3}\Delta(\stv),
\end{equation*}
proving the lemma.
\end{proof}

If we consider that every machine uses {\apr}, then, as a result of Lemma~\ref{lem:beta} along 
with Lemma~\ref{lem:basic} and the fact that $\Capr(\stv)=2\Csmith(\stv)$ for any configuration 
$\stv$, we get the following theorem bounding the running time of our algorithm.

\begin{theorem}\label{thm:alg}
Starting from any initial configuration $\stv^0$ and following basic dynamics leads to 
a profile $\stv$ with $\Csmith(\stv)\leq (2+O(\epsilon))\OPT$ in at most $O\left(\frac{n}{\epsilon}\log{\left(\frac{\PotA(\stv^0)}{\PotA(\stv*)}\right)}\right)$ steps.
\end{theorem}

\section{Concluding remarks}
On mapping machines to edges of a parallel link network, the machine
scheduling
problem for the case of related machines becomes a special
case of general
selfish routing games. In this context, the ordering policies on
machines correspond
to local queuing policies at the edges of the network. From this
perspective, it would be
interesting to generalize our results to network routing games.
Designing such local
queuing policies would be an important step toward more realistic models of
selfish routing
games when the routing happens over time~\cite{HMRS09, FOV08, KS11}. We hope that our new
technique along with the policies proposed in this paper could serve
as a building block
toward this challenging problem.

All the mechanisms discussed here are strongly local. For the case of
the makespan
objective, one can improve the approximation ratio from $\Theta(m)$ to
$\Theta(\log m)$ by using local policies instead of just strongly
local policies. It remains
open whether there are local policies that perform even better than our
strongly
local ones.

%

%

\paragraph{Acknowledgements.} We thank Tanmoy Chakraborty for helpful discussions and Ioannis Caragiannis for pointing out related literature. The fourth author also thanks Yossi Azar for initial discussions about the subject of study of this paper. Part of this work was done while the second and last authors were visiting EPFL; we thank Fritz Eisenbrand for his hospitality. 

\bibliography{schedule}
\bibliographystyle{plain}

\appendix

\section{Lower Bounds}\label{appendix:lowerbound}
\paragraph*{Deterministic non-preemptive strongly local coordination mechanisms.}
Adapting a proof of Caragiannis et al.~\cite[Theorem 7]{CFKKM06}, Correa and Queyranne~\cite[Theorem 13]{CQ10}, showed that the pure PoA of games induced by {\smi} can be arbitrarily close to $4$.
This holds even with unweighted jobs in the restricted identical machines model 
(a similar construction in~\cite{FOV08} for nonatomic prioritized selfish routing can also be easily adapted).
Based on this construction, we show the same lowerbound for arbitrary strongly local prompt policies that are deterministic and non-preemptive.
To aid in comprehension, we first demonstrate this for strongly local ordering policies (i.e., deterministic non-preemptive policies where the IIA property holds).
We then discuss the changes needed to obtain the full result, but since the arguments are of a quite different flavour from the rest of the paper, we give only a sketch.
\begin{proof}[Proof of Theorem \ref{thm:lowb}]
We begin by presenting the family of game instances that leads to pure PoA approaching 4 for games
induced by {\smi} in the restricted identical machines model~\cite{CQ10}, and then show how to generalize the lowerbound based on this construction.

There are $m$ machines and $k$ groups of jobs $g_1,\ldots, g_k$, where group $g_x$ has $m/x^2$ jobs.
We assume that $m$ is such that all groups have integer size and let $j_{x y}$ denote the $y$-th job
of the $x$-th group. A job $j_{x y}$ can be assigned to machines
$1,\ldots,y$, and we assume that for two jobs $j_{xy}$ and $j_{x'y'}$ with $y < y'$, then $j_{x'y'}$ has higher priority than $j_{xy}$ (if $y-y'$, the ordering can be arbitrary).

If every job $j_{x y}$ is assigned to machine $y$, there are exactly $m/x^2$ jobs
with completion time $x$ ($1\leq x \leq k$), which leads to a total cost of $m\sum_{x=1}^k{1/x}$. 
On the other hand, assigning each job to the machine with
smallest index among all the ones that minimize its completion time gives a PNE assignment
whose total cost is $\Omega(4m\sum_{x=1}^k{1/x})$~\cite{CQ10}.

\sparagraph{Ordering policies.}
We may of course modify the construction so that each job $i$ can be assigned to only two
machines: the machine $O_j$ to which it is assigned under \OPT, and the machine $N_j$ it assigned to under the Nash
(where $N_j\leq O_j$ for all $j$). Since the job ordering under the optimal assignment does not affect the cost, 
we only need to make sure that for any jobs $j$, $j'$ with $O_{j'} < O_j$, $j$ gets higher priority than $j'$ on $O_j$.

Given a specific lower bound instance for {\smi}, we have $n$ job \emph{slots}, each 
defined by the pair of machines $O_j$ and $N_j$. Given a set of ordering policies, each machine has its own
strictly ordered list of all $n$ jobs. What we need to do is assign a specific job to each
slot so that the ordering restrictions as specified in the previous paragraph
comply with the lists. We start from
the slot $j$ with the greatest $N_j$ machine index and we assign the first job of machine
$N_j$'s list to this slot. We then erase this job from all lists and repeat. In case of a tie,
that is if there is more than one slot with the same $N_j$, we first consider the slots with
greater $O_j$ machine index. This ensures that, given the PNE assignment, any job that 
deviates back to its OPT machine will suffer cost at least as much as in the {\smi} instance,
while its cost in the PNE is the same as in the given instance. Therefore, the assignment of each job
$j$ to machine $N_j$ is a PNE for this set of ordering policies.

\sparagraph{Removing the IIA assumption.}
We modify the above construction to have $N$ jobs, where $N$ is extremely large, and one extra machine (so we have $M=m+1$ machines). 
Each machine has an associated prompt policy (which may use job IDs); thus for any subset of jobs, the policy on a machine will specify the order that the jobs are run.
We will then choose only a small subset of $n$ jobs that will fill in the previously defined slots; the remaining jobs will all be assigned processing time $0$ on machine $m+1$, and infinity on all other machines; call such jobs \emph{spurious}.
By choosing the assignment of jobs to slots appropriately, we will be able to enforce the orderings we want on the jobs, and obtain a Nash with the same cost as before.
More precisely, we want the following, which ensures that the proposed Nash assignment is indeed an equilibrium:
\begin{enumerate}[(i)]
	\item In the Nash assignment, the ordering on any machine is exactly as we require in the previously defined construction.
	\item If we take the Nash assignment, but then any single job attempts to deviate, it will find itself at the back of the ordering. 
		More carefully: if $S_i$ is the set of jobs on machine $i$ at Nash, and we consider any job $j$ with $O_j = i$, then $j$ is last according to the ordering determined by the set $S_i \cup \{j\}$ and the policy on machine $i$.
		This ensures that nobody has an incentive to deviate.
\end{enumerate}
To prove this, we begin with the $m$-th machine, and argue that we can find a set $S_m \subset J$, with $|S_m|$ equal to the number of slots which have machine $m$ as the Nash strategy, such that there is a very large set $Q_m \subset J$ with the following property:
\[ \text{Every job $j \in Q_m$ is last in the ordering on machine $m$ determined by the set $S_m \cup \{j\}$.}\]
We will assign $S_m$ to the slots which run on machine $m$ at \OPT, and then make all jobs outside of $S_m$ and $J_m$ spurious. 
We then repeat this process on machine $m-1$, but selecting $S_{m-1}$ and $Q_{m-1}$ as subsets of $Q_m$. 
This construction guarantees an ordering satisfying properties (i) and (ii).
The existence of the sets $S_i$ and $Q_i$ for all $i$ follows from the following easily proved combinatorial lemma, asuming that $N$ is chosen sufficiently large.
\begin{lemma}
	 Let $k$ and $r$ be integers, with $k > r$. 
	 For any subset $S$ of $[k] := \{1,2,..., k\}$, let $\pi_S$ be an ordering (permutation) of $S$, which may depend on $S$ in an arbitrary manner, and define
	 \[ Q_S := \{ j \in [k] \setminus S: j \text{ is last according to the order } \pi_{S \cup \{j\}}  \}. \]
Then there exists a subset $S$ of size $r$ so that $|Q_S| \geq (k-r)/(r+1)$.
\end{lemma}
\end{proof}

\paragraph*{Deterministic strongly local mechanisms.}
We give here a lower bound that applies to \emph{any} deterministic prompt strongly local coordination mechanism, even when preemption is allowed, as long as jobs are anonymous.
\begin{proof}[Proof of Proposition~\ref{prop:lbdet}]
The construction is a slight variant of one given in Caragiannis et al.~\cite{CFKKM06} for load balancing games. We define the construction in terms of the \emph{game graph}; a directed graph, with nodes corresponding to machines, and arcs corresponding to jobs. The interpretation of an arc $(i^*,i)$ is that the corresponding machine is run on $i$ at the Nash equilibrium, and $i^*$ in the optimal solution (all jobs can only be run on at most two machines in the instance we construct).

Our graph consists of a binary tree of depth $\ell$, with a path of length $\ell$ appended to each leaf of the tree. In addition, there is a loop at the endpoint of each path. All arcs are directed towards the root; the root is considered to be at depth zero. In the binary tree, on a machine at depth $i$, the processing time of any job that can run on that machine is $(3/2)^{\ell-i}$. In the chain, on a machine at distance $k$ from the tree leaves all processing times are $(1/2)^k$.

By slightly perturbing the processing times of jobs on different machines it is easily checked that if every job is run on the machine pointed to by its corresponding arc, the assignment is a pure NE. The latter holds for arbitrary prompt strongly local coordination mechanisms so long as jobs are anonymous. On the other hand, if all jobs choose their alternative strategy, we obtain the optimal solution. A straightforward calculation  shows that, in the limit $\ell \rightarrow \infty$, the ratio of the cost of the NE to the optimal cost converges to $13/6>2.166$.
\end{proof}

\paragraph*{\rand.} The previous instance can be easily modified to give a lower bound on the performance of \rand. Just take the same instance but replace $3/2$ by $4/3$ and $1/2$ by $2/3$. The same assignment then gives a PNE, and in this case the ratio of interest approaches $5/3$.



\section{The performance of $\mathbf{Rand}$ on a single machine}\label{appendix:randsingle}
\begin{proof}[Proof of Theorem~\ref{thm:chung}]

We want to prove that for any sequence $u_1,\ldots,u_k$, $u_i \geq 0$, the following inequality holds:
\[
\sum_i\sum_j u_iu_j \frac{ij}{i + j} \leq \frac{\pi}{4}  \sum_i\sum_j u_iu_j \min\{i,j\}.
\]
We will in fact prove that for any sequence $x_1, x_2, \ldots, x_n$, $x_i \in \N$, 
\begin{equation}\label{eq:hilberty}
\sum_i\sum_j \frac{x_ix_j}{x_i+x_j} <  \frac{\pi}{4}  \sum_i\sum_j \min\{x_i,x_j\}.
\end{equation}
This implies the inequality in the statement, for the choice
$u_r = |\{i: x_i =  r\}|$, and hence clearly for any integer sequence $(u_i)$. An obvious scaling argument then gives it for general nonnegative $u_i$.

Since both summations in~\eqref{eq:hilberty} are symmetric, we may assume without loss of generality that $x_1\ge \cdots \ge x_n\ge 0$. Then, we note that $\sum_{i=1}^n\sum_{j=1}^n \min\{x_i,x_j\}= 2 \sum_{i=1}^n x_i(i-1/2)$. Also, observe that the inequality is homogeneous so that proving the inequality is equivalent to proving that the optimal value of the following concave optimization problem is less than $\pi/2$:
\[
z= \max\left\{\sum_{i=1}^n\sum_{j=1}^n \frac{x_ix_j}{x_i+x_j}:\text{ s.t. }
\sum_{i=1}^n x_i(i-1/2)=1,\,  x_1\ge \cdots \ge x_n\ge 0 \right\}.
\]
Clearly $z\le z'$, where 
\[
z'= \max\left\{\sum_{i=1}^n\sum_{j=1}^n \frac{x_ix_j}{x_i+x_j}:\text{ s.t. }
\sum_{i=1}^n x_i(i-1/2)=1,\,  x_i \ge 0 \text{ for all }i=1,\ldots,n \right\}.
\]
Furthermore, we may assume that in an optimal solution all variables satisfy $x_i>0$. 
Otherwise, we could consider the problem in a smaller dimension. 
Thus, the KKT optimality conditions state that for all $i=1,\ldots,n$ we have
\begin{equation}
	\mu (i-1/2) = 2 \sum_{j=1}^n \left(\frac{x_j}{x_i+x_j}\right)^2. \label{eq:kkt}
\end{equation}
Multiplying by $x_i$, summing over all $i$, and using $\sum_{i=1}^n x_i(i-1/2)=1$, we obtain:
\[
\mu = 2 \sum_{i=1}^n\sum_{j=1}^n x_i\left(\frac{x_j}{x_i+x_j}\right)^2
=\sum_{i=1}^n\sum_{j=1}^n \frac{x_ix_j}{(x_i+x_j)^2}(x_i+x_j)=z'.
\]
Now consider~\eqref{eq:kkt} with $i^*=\arg\max_{i} x_i(i-1/2)^2$. We have that
\[
z' = \frac{2}{i^*-1/2} \sum_{j=1}^n \left(\frac{x_j}{x_{i^*}+x_j}\right)^2 \le 
2(i^*-1/2)^3 \sum_{j=1}^\infty \left(\frac{1}{(i^*-1/2)^2+(j-1/2)^2}\right)^2.
\]
Using standard complex analysis it can be shown that the latter summation equals 
\[ (\pi/2)\bigl((i^*-1/2) \pi \tanh(\pi (i^*-1/2))^2+ \tanh(\pi (i^*-1/2))-\pi (i^*-1/2)\bigr), \]
which is less than $\pi/2$.
\end{proof} 


\section{A Reduction from Prioritized Selfish Routing}\label{appendix:reduction}
In this appendix, we show that in the unweighted case, and using \shf, the scheduling games under consideration form a special cases of the priority selfish routing games defined in~\cite{FOV08}. This suffices to give upper bounds on the \poa\ for \shf, and in fact the correct bound if the nonatomic case is considered. 

A \emph{priority selfish routing game} is defined as follows (except here, we will restrict ourselves to the unweighted case, where all players have unit demand).
We are given a directed graph $G$, and a set of players $j=1,\ldots, n$; each player has an associated source-sink pair $(s_j, t_j)$, and must pick as their strategy some $s_j$-$t_j$ path to route their demand.

Each arc $e$ has an associated cost function $f_e$, which we will take to be linear; $f_e(x) = a_ex + b_e$.
In the standard selfish routing game, the cost or delay experienced by a player using edge $e$ is given by $f_e(x_e)$, where $x_e$ is the load on the edge, i.e.~(taking unit demands) the number of players using that edge.
The total cost associated to an edge is then $x_ef_e(x_e)$.
In the priority selfish routing model on the other hand, the total cost for an edge will be $\int_0^{x_e} f_e(z)dz$, the ``area under the curve''.
This is split between the players using an edge, according to some ordering $\prec_e$ of the players using edge $e$:
The $t$'th player in the ordering pays an amount $\int_{t-1}^t f_e(x)dx$.
The ordering $\prec_e$ can be very general, and may depend on the strategies chosen by all the players (even those not using edge $e$).

In~\cite{FOV08}, it is shown that in this model, the price of anarchy is at most $17/3$ in the setting described above, which improves to $4$ in the nonatomic case where any individual player is negligible.
These upper bounds hold for any priority ordering.

We are now ready to describe the reduction.
Begin with an instance of the scheduling game, with policy given by \shf.
By scaling if necessary, assume that all finite $\p{i,j}$ satisfy $\p{i,j} \leq 1$, and let $Q$ be such that $Q\cdot \p{i,j} \in \N$ for all finite $\p{i,j}$.
We construct a graph $G$ as follows. There is a single sink node $t$ which will be the destination for all players.
Each machine $i$ will correspond to a path $P_i$ of length $Q$, and the cost function of each edge on the path will be simply $f_e(x) = x / Q$.
Connect the end of each path to a common sink node $t$, with zero cost edges.

Now for each job $j$, we will have a source node $s_j$, and a player with source $s_j$ and destination $t$.
For each machine $i$, we add an arc from $s_j$ to a node $v_{ij}$ in the path corresponding to $i$, such that the fraction of the path between $v_{ij}$ and the end of the path (towards $t$) is exactly $\p{i,j}$.
The cost of this arc will be a constant $\p{i,j}/2$.
To complete the definition of the priority selfish routing instance, we define the priority ordering on any edge in $P_i$ according to \shf, in increasing order of $\p{i,j}$.

There is a natural correspondence between an assignment in the scheduling problem and a routing in the priority routing problem. If job $j$ uses machine $i$, then route $j$ from $s_j$ to $v_{ij}$ and then to $t$.

\begin{lemma}
  For any job $j$, the completion time $\ct_j$ of the job in the scheduling instance is the same as the amount $C_j$ player $j$ pays in the derived priority routing instance.
\end{lemma}
\begin{proof}
Suppose job $j$ uses machine $i$. 
In the routing instance, all larger (w.r.t.\ processing time) jobs on the edge will not affect job $j$'s cost, since shorter jobs have higher priority,
All smaller jobs on the other hand will cause delays, on some subset of the edges on the path.
In particular, a job $k$ with $\p{i,k} < \p{i,j}$ causes a delay of $1$ on a fraction $\p{i,k}$ of the edges on machine $i$'s path.
On an edge with $\ell$ players ahead of player $j$, $j$ will pay an amount given by a trapezoidal area: $\tfrac1{2Q}(\ell + (\ell+1)) = \tfrac{\ell+1/2}{Q}$.
Summing up the costs over all edges used by $j$, we get
\begin{equation*}
  C_j = \p{i,j}/2 + \sum_{k: \p{i,k} < \p{i,j}} \p{i,k}\, +\, \p{i,j}/2
  = \p{i,j} + \sum_{k: \p{i,k} < \p{i,j}} \p{i,k}. \qedhere
\end{equation*}
\end{proof}

It follows immediately that the Nash equilibria of the scheduling game and the derived priority routing coincide, and that social costs are also the same.
Thus the worst-case \poa\ of the unweighted scheduling game is no worse than the worst-case price of anarchy in the unweighted priority routing model, i.e., $17/3$.

\end{document}